\documentclass[letterpaper, 10 pt, conference]{ieeeconf}  

\IEEEoverridecommandlockouts                              
\usepackage{textcomp}
\usepackage{amsmath,amssymb}
\usepackage[margin=1in]{geometry}
\usepackage{units}
\usepackage{mathtools}
\usepackage{float}
\usepackage{csquotes}
\usepackage{xcolor}
\usepackage{dsfont}
\newtheorem{definition}{Definition}
\newtheorem{assumption}{Assumption}
\newtheorem{theorem}{Theorem}
\newtheorem{corollary}{Corollary}
\newtheorem{lemma}{Lemma}
\newtheorem{example}{Example}

\DeclarePairedDelimiter\set{\{}{\}}%
\DeclarePairedDelimiter\abs{\lvert}{\rvert}%

\newcommand{\C}{\mathcal{C}}

\newcommand{\A}{\mathcal{A}}
\newcommand{\G}{\mathcal{G}}

\usepackage{xcolor}
\RequirePackage{url}
\usepackage[T1]{fontenc}
\usepackage[pdfa, hidelinks]{hyperref}
\hypersetup{      
    colorlinks=true,                
    breaklinks=true,                
    urlcolor= black,                
    linkcolor= blue,                
    bookmarksopen=false,
    filecolor=black,
    citecolor=blue,
    linkbordercolor=blue
}
\usepackage{soul}
\usepackage{lipsum}
\usepackage{bbm}
\usepackage{tikz}
\title{\LARGE \bf
Towards Assume-Guarantee Profiles for Autonomous Vehicles
}
\author{Tung Phan-Minh$^*$, Karena X. Cai$^*$, Richard M. Murray
\thanks{$^*$ Equal contribution}
\thanks{Tung Phan-Minh is a graduate student in Mechanical Engineering at the California Institute of Technology
       {\tt\small tmphan at caltech dot edu}}%
\thanks{Karena Cai is a graduate student in Control and Dynamical Systems at the California Institute of Technology
       {\tt\small kcai at caltech dot edu}}
\thanks{We would like to acknowledge DENSO and VeHiCal for funding this project and Soon-Jo Chung and Andrea Censi for the helpful discussions and feedback.}
}

\begin{document}
\maketitle
\thispagestyle{empty}
\pagestyle{empty}

\begin{abstract}
Rules or specifications for autonomous vehicles are currently formulated on a case-by-case basis, and put together in a rather ad-hoc fashion. As a step towards eliminating this practice, we propose a systematic procedure for generating a set of supervisory specifications for self-driving cars that are 1) associated with a distributed assume-guarantee structure and 2) characterizable by the notion of consistency and completeness. Besides helping autonomous vehicles make better decisions on the road, the assume-guarantee contract structure also helps address the notion of blame when undesirable events occur. We give several game-theoretic examples to demonstrate applicability of our framework.
\end{abstract}

\section{Introduction}
In the near future, autonomous vehicles will likely have to function alongside human-operated vehicles, pedestrians, cyclists, and more---that is, until a fully-automated transportation infrastructure can be built. The interaction between self-driving cars and humans will inevitably result in accidents.

Self-driving car manufacturers are therefore responsible for designing the high-level behavior of cars such that they minimize the risk of collision. Currently, however, the rules that self-driving cars follow are often designed heuristically, and are therefore lacking in transparency, predictability, and performance \cite{Survey2016, 2017FormalModel}.

We argue that if all self-driving cars were to adhere to some behavioral contract, there would be much greater certainty of how other self-driving cars would behave, thereby making it significantly easier to choose actions that would be mutually beneficial and also reduce the risk of accidents. The process of identifying the party responsible for an accident would also be simplified. Furthermore, car manufacturers could even begin to optimize their car behavior to accommodate driving preferences in addition to safety, like courtesy or fuel efficiency. 


To describe this contract, formal methods appear to be a good place to start. These provide many tools and formalisms for tackling specifications that guarantee high-level behaviors like safety and liveness in complex systems like self-driving cars \cite{modelChecking}. Temporal logic specifications oftentimes, however, rely on impractical assumptions on the environment to reduce the search space \cite{correctCont2011}. Additionally, the set of specifications for self-driving cars are often scenario-specific and are formulated independently of one another \cite{2017FormalModel}.

A more ontological approach to handling specifications involves designing ``rulebooks'' that specify the high-level behaviors of self-driving cars \cite{2019Rulebooks}. The rulebooks hierarchically order the set of rules for self-driving cars via a preorder. The preorder intentionally leaves ambiguity in how the car will choose to follow the rules, and therefore does not admit well-defined car behavior. The authors in \cite{2019Rulebooks} cannot make any guarantees about the correctness of car behavior, because they do not make any explicit assumptions about how agents in the environment behave. In particular, they do not address how to accommodate for the unpredictable and law-evasive nature of human drivers. 

In light of these issues, we propose a framework that can be used to:
\begin{enumerate}
\item Identify high-level specifications and their relations to one another as part of a hierarchical structure that helps self-driving cars achieve desirable behavior on the road.
\item Define what it means for a set of specifications to be consistent and complete (drawing inspiration from formal methods). 
\item Introduce an assume-guarantee contract formalism for specification structures as well as notions of rationality and blame.
\item Present a basic and consistent set of axioms for self-driving cars that can be refined and built upon.
\item Demonstrate with game-theoretic examples how rational autonomous vehicle behaviors can be computed/agreed upon under the assumption they are aware of each other's specification structures.
\end{enumerate}

We ultimately want to be able to guarantee that self-driving cars will behave correctly and not be responsible for accidents. This paper offers a step in that direction.

\section{Overview}
In a dynamic and interactive environment, the problem of providing guarantees for a single agent without making any assumption on the behaviors of other agents is ill-posed. We show the inherent coupling between the assumptions on the environment and the system's guarantees in Fig.~\ref{highlevel}. We propose the framework of assume-guarantee profiles to explicitly address this issue.
\begin{definition}[Assume-guarantee profile]
\label{ag_profile}
An \textit{assume-guarantee profile} for an agent is a 2-tuple  $(\mathcal{A}, \mathcal{G})$ where
\begin{itemize}
    \item $\mathcal{A}$ is a set of behavioral preferences or characteristics that the agent assumes its environment to have. 
    \item $\mathcal{G}$ is a set of behavioral preferences or characteristics that it is obligated to behave according to as long as its environment makes decisions in accordance with $\mathcal{A}$.
\end{itemize}
\end{definition}
   \begin{figure}
      \centering
      \includegraphics[scale=0.33]{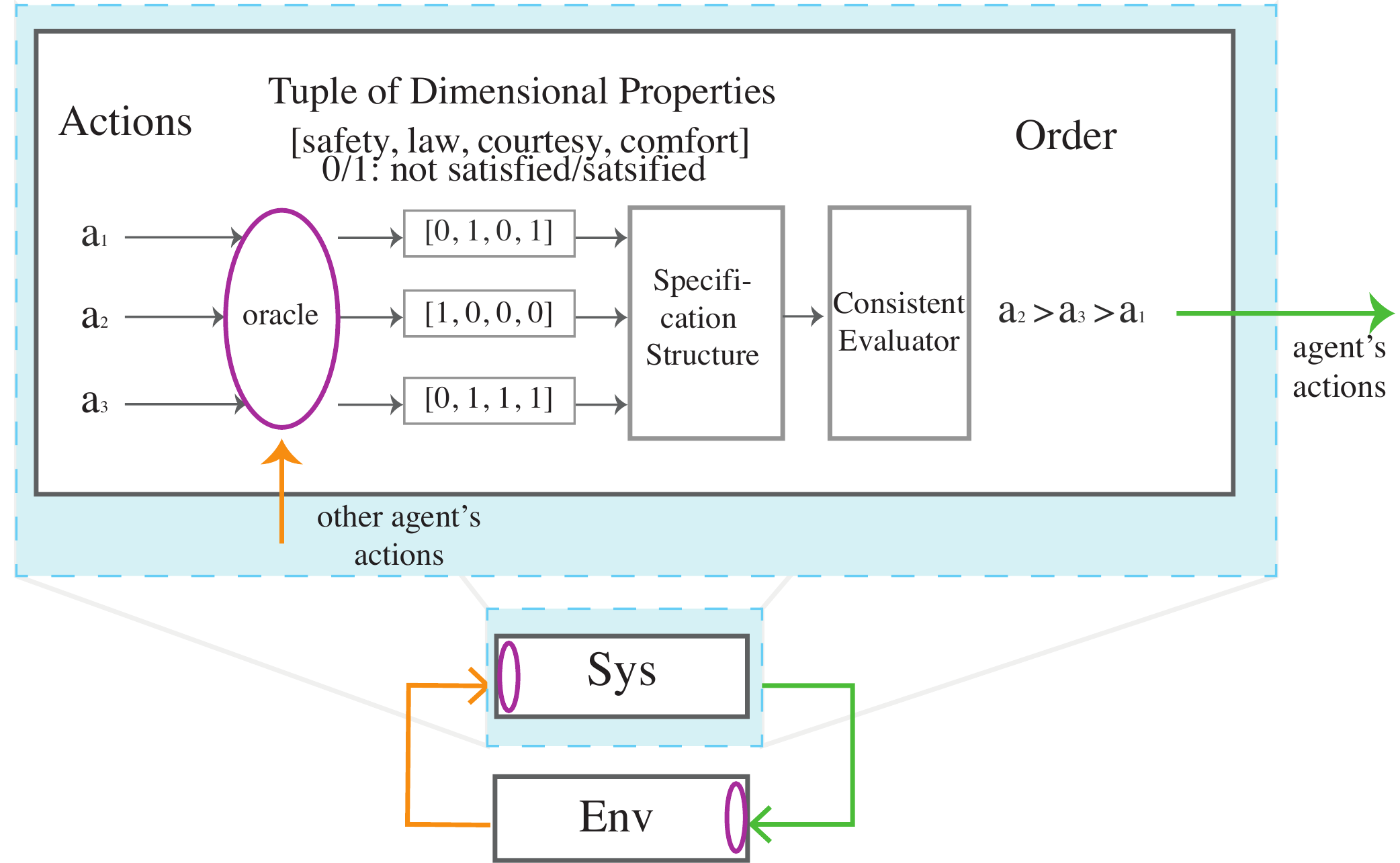}
      \caption{A high-level system architecture capturing the inherent coupling of the behavioral specifications for an agent and its environment is shown in the bottom figure. Each agent identifies the best action to take based on which subset of rules are satisfied by that action. The specification structure and conisistent evaluating function define a unique ordering on all subsets of rules that ultimately determine which actions are better or worse than others. }
      \label{highlevel}
   \end{figure}

To model the sets of behavioral preferences or characteristics mentioned in Definition~\ref{ag_profile},
we propose a mathematical object termed a \textit{specification structure} that imposes a hierarchy on sets of what we call \textit{dimensional properties}. A property is a desirable attribute that must be satisfied or not satisfied. A dimensional property is a property whose satisfaction is independent of the satisfaction (but not dissatisfaction!) of other properties. Examples of dimensional properties include safety, lawfulness, courtesy, and comfort. Safety does not imply comfort but being unsafe can imply discomfort. The guarantee we make in our assume-guarantee contract is that the self-driving car will act in accordance with the specification structure. Intuitively, this means of all subsets of specifications that the car can satisfy, it will choose to satisfy the one that is ranked highest in priority with respect to the specification structure. We define this more rigorously in the next section.

\section{Evaluator and Evaluated Structure}
Before presenting the formal definition of a specification structure, we make the following assumption about the predictive capabilities of a self-driving car.

\begin{assumption}[Oracle]
\label{oracle}
We assume that each autonomous agent relies on an \textit{oracle} \cite{sipser2012introduction} that provides predictions to its queries about the satisfaction of dimensional properties of interest for any action or strategy that it is considering. The input to the oracle is a set of dimensional properties, a potential action or strategy the car can choose to take, and the current world state configuration. The output of the oracle is some prediction of what specifications (dimensional properties) will be satisfied if a strategy is followed. In the most simple case, the oracle could return a valuation of a set of a Boolean variables, each indicating whether or not a property is violated. 
\end{assumption}
Although many decision/optimization problems currently posed for autonomous vehicles are of high computational complexities, not to mention undecidable \cite{madani2003undecidability, papadimitriou1987complexity}, we expect future technology to be capable of approximating the oracle to an acceptable level of fidelity (see~\cite{lefevre2014survey} for a sample of related methods).


If a set of dimensional properties is simply partially ordered, then there may not be not enough structure to uniquely identify which action should be taken.


\begin{example}
Consider a set $S = \{a, b, c, d, e\}$ that is partially ordered (a poset) such that $b \prec a$, $c \prec a$, $d \prec c$, and $e \prec c$. Here, each element in the set represents a dimensional property like safety, the law, performance, etc. By this partial order, the node $b$ cannot be compared to $c$ or $d$ or $e$. For a self-driving car, any action will result in satisfying a subset of the dimensional properties. Since $b$ cannot be compared to $c$ or $d$ or $e$, it is ambiguous whether a self-driving car should take an action that satisfies the properties $a$, $b$, and $d$ or an action that satisfies $a$, $b$, and $e$.
\end{example}
\medskip\noindent 
With only a partial ordering on the set of dimensional properties, there can be ambiguity in choosing the best subset of properties that can be satisfied. Our aim in this paper is to define a function evaluator over a set of dimensional properties such that it resolves all ambiguity and admits only a unique way to interpret the ordering of subsets relative to one another. Furthermore, if there exists a partial order on the set of dimensional properties, the function evaluator must respect that partial order. 
We therefore introduce the idea of \textit{consistent evaluators}, which are a class of functions that can endow some posets (in our case, dimesional properties) with a unique \textit{weak} order on their powersets.  Being weakly ordered means that all subsets are comparable but some subsets may have equal values to each other (these are considered indistinguishable). 

In a practical setting, if a self-driving car manufacturer wanted to impose a total order instead of a weak order on the powerset, they would have to face the challenging task of defining how any one set of dimensional properties is strictly better or worse than another set of dimensional properties. This is arguably not only impractical because of the exponential growth in the size of the powerset, but also because sometimes a \textit{strict} comparison among sets of properties is simply unnecessary. A consistent evaluator, which allows for sets in the powerset to have equal value, therefore allows for a more sensible way of resolving comparisons between subsets of specfications.
 

We refer to maximal chains (antichains) of partially-ordered sets in our definitions and proofs, so we present the definitions here. 
\begin{definition}[Maximal Chain (Antichain)]
A chain (antichain) is a subset of a partially ordered set such that any two distinct elements in the subset are comparable (incomparable). A chain (antichain) is maximal when it is not a proper subset of another chain (antichain). 
\end{definition}

\begin{definition}[Consistent evaluator]
\label{evaluator}
Given a set of dimensional properties $\mathcal{P}$ and its powerset $2^{\mathcal{P}}$, we say that $f: 2^{\mathcal{P}} \to T$, where $T$ is a totally ordered set with $\leq$ as the ordering relation, is a consistent evaluator for $\mathcal{P}$ if for all subsets $P_1,P_2 \subseteq \mathcal{P}$, the following must hold: 
\begin{enumerate}
    \item\label{worstbest} $P_1 \neq \varnothing \Rightarrow f(\varnothing) < f(P_1)$
        \item\label{evalremove} $p_1 \in P_1 \land p_2\in P_2 \land f(\set{p_1})=f(\set{p_2}) \Rightarrow (f(P_1) \leq f(P_2) \Rightarrow f(P_1-\set{p_1}) \leq f(P_2-\set{p_2}))$ 
    \item\label{evaldiscrim} $(\forall p_1 \in P_1. \forall p_2 \in P_2. f(\set{p_1}) \neq f(\set{p_2})) \Rightarrow (  \max\limits_{p \in P_1} f(\set{p}) < \max\limits_{q \in P_2} f(\set{q}) \Rightarrow f(P_1) < f(P_2))$
\end{enumerate}
If $\mathcal{P}$ is partially ordered by $\preceq$ and $\mathfrak{A}_{(P,\preceq)}$ is the set of all antichains of $P$, we further require that for any $p_1, p_2 \in P$
\begin{enumerate}
  \setcounter{enumi}{3}
      \item\label{evalorder} $p_1 \prec p_2 \Rightarrow f(\set{p_1}) < f(\set{p_2})$
      \item\label{evalcaution} $(\set{p_1, p_2} \in \mathfrak{A}_{(P, \preceq)} \land f(\set{p_1}) < f(\set{p_2})) \Rightarrow \exists s,t \in P. p_1 \prec s \land f(\set{s}) = f(\set{p_2}) \land  f(\set{p_1}) = f(\set{t}) \land t \prec p_2$
\end{enumerate}
\end{definition}
Intuitively, the conditions in Definition~\ref{evaluator} mean
\begin{enumerate}
    \item  The evaluator will assign the worst value when no property is satisfied. This ensures that every property included in $\mathcal{P}$ matters to the evaluator.
    \item Properties of equal value to the evaluator can be disregarded without affecting the result of the evaluation.
    \item For sets that do not have properties with the same values, the one with the most highly valued property is preferable.
    \item If there exists a pre-imposed hierarchy between some of the properties, then the evaluator must respect it.
    \item Given a pre-imposed hierarchy on the properties, the evaluator must be impartial: it will only assign different values to two properties whose relationship is not defined in the hierarchy when they are comparable via two equally valued ``proxies''. 
\end{enumerate}
\begin{example}
Consider a partially ordered set $Q$ in which $p$ is the greatest element and all other elements belong to an antichain. Then we can define $f$ as the function $f(\tilde{Q}) \coloneqq \mathds{1}_{p \in \tilde{Q}} \abs{Q} + \abs{\tilde{Q}-\set{p}}$ for all $\tilde{Q} \subseteq Q$ where $\mathds{1}$ is the indicator function. This function evaluates any subset with the maximal element in it as the cardinality of $Q$ plus the dimension of the subset not including the element $p$. It also evaluates any subset without the maximal element as the dimension of that subset. One can easily verify that $f$ is a consistent evaluator for $Q$.
\end{example}
   \begin{figure}[H]
      \centering
      \includegraphics[scale=0.25]{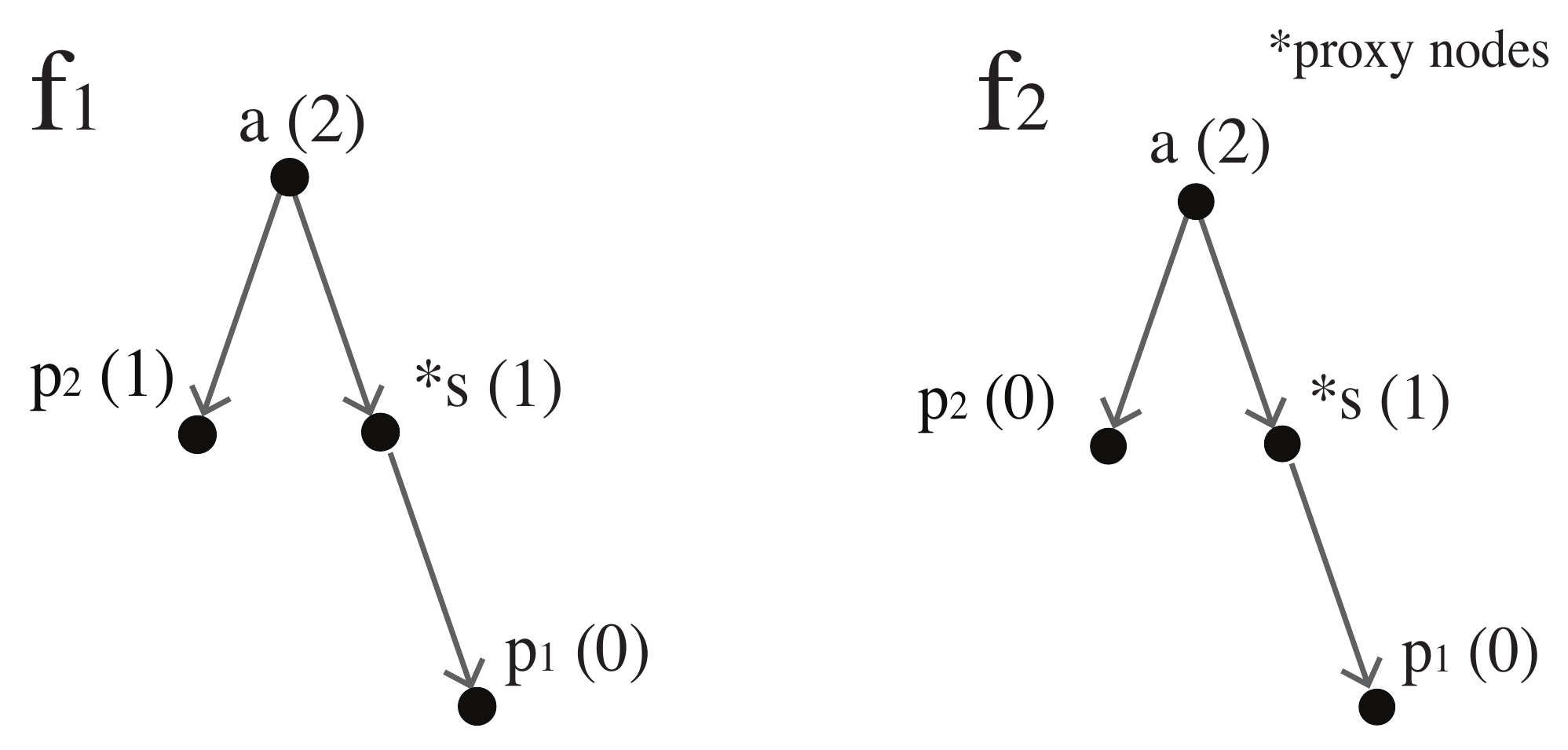}
      \caption{A poset that does not admit a consistent evaluator. The values in parentheses denote the value of the singleton set containing that node given by the evaluator $f_i$. In both cases, requirement~\ref{evalcaution} is violated.}
      \label{inconsistent}
   \end{figure}
\begin{example}[Poset without consistent evaluator]
Consider the poset with the structure shown in 
Fig.~\ref{inconsistent}. We cannot define a consistent evaluator that satisfies all five requirements on this partially-ordered set. In order to satisfy requirement~\ref{evalorder}, such that the 
partial order established in the poset is preserved, the 
consistent evaluator function $f_1$ on the left and $f_2$ on 
the right can, WLOG, assign all nodes on the right branch of the poset with the values shown in Fig.~\ref{inconsistent}. To respect the partial order, by requirements~\ref{evalorder} and \ref{evalcaution} of consistent evaluators, $p_2$ must be assigned a value in $\set{0,1}$. The left figure shows what happens if the function takes on the value $1$ for the left node, i.e., $f_1(\{p_2\}) = 1$. If this happens, then $f_1(\{p_1\}) < f_1(\{p_2\})$ but there is no proxy node that is comparable to $p_2$ in the poset and has a value equal to $f_1(\{p_1\})$. This clearly violates requirement~\ref{evalcaution}.
The right figure shows what happens if the function takes on the
value $0$, i.e. $f_2(\{p_1\}) = 0$. A similar violation is incurred by $f_2$ (see the Appendix).
\end{example}
\medskip\noindent 

Through the previous examples, we see that not all posets admit a consistent evaluator. The natural question to ask, is: what makes posets consistently evaluable? The next theorem will answer this question. 
\begin{theorem}
\label{consist_evaluable}
A finite poset $P$ of dimensional properties has a consistent evaluator if and only if it can be partitioned by a set $\A$ of $N$ maximal antichains such that
\begin{enumerate}
    \item\label{rank_prop} Maximal Antichain and Rank Criterion: The maximal antichains $\A$ can be assigned ranks in such a way that the partial order is respected. 
    \item\label{path_prop} The Maximal Chain Criterion: For each node (dimensional property), there exists a maximal chain containing the node of length $N$.
\end{enumerate}
\end{theorem}
\renewcommand{\proof}{Proof:}

We are ready to give a proof sketch for Theorem~\ref{consist_evaluable}.

\renewcommand{\proof}{Proof (sketch):}
\begin{proof}
($\Rightarrow$):
We associate all properties that have the same value with a unique antichain. This set of antichains can be shown to form the partition that has properties~\ref{rank_prop}) and ~\ref{path_prop}).  The reader is referred to the Appendix for a full proof.
\\
($\Leftarrow$): 
We prove this direction by construction. If a partially-ordered set is partitioned into ranked maximal anti-chains, we can construct a function where a subset is evaluated based on counting the number of nodes that are satisfied in every rank. Two subsets are comparable via this function by lexicographical ordering, where the number of nodes counted for higher ranking antichains are counted as more significant. We then show that this function satisfies all the properties defined for a consistent evaluating function. Readers can refer to \ref{W_func_spec} for a clarifying example and the Appendix for the full proof.  

\end{proof}
\renewcommand{\proof}{Proof:}
Is it possible that there may be multiple such decompositions of maximal antichains, making the ordering that is induced via the corresponding rankings is not unique and hence the consistent evaluator not very consistent? Luckily, the answer is a reassuring negative.
\begin{theorem}
\label{uniqueness}
Such a partition in Theorem~\ref{consist_evaluable} is unique.
\end{theorem}
\renewcommand{\proof}{Proof (sketch):}
\begin{proof}
This result can be obtained using the pigeonhole principle and property~\ref{path_prop} of Theorem~\ref{consist_evaluable}. A full argument is available in the Appendix.
\end{proof}
We have defined the necessary and sufficient properties posets need to have so they can be consistently evaluated. Since this set of posets is not super intuitive, we introduce graded posets, which we show are also consistently-evaluable but easier to reason about. 

   \begin{figure}[thpb]
      \centering
      \includegraphics[scale=0.28]{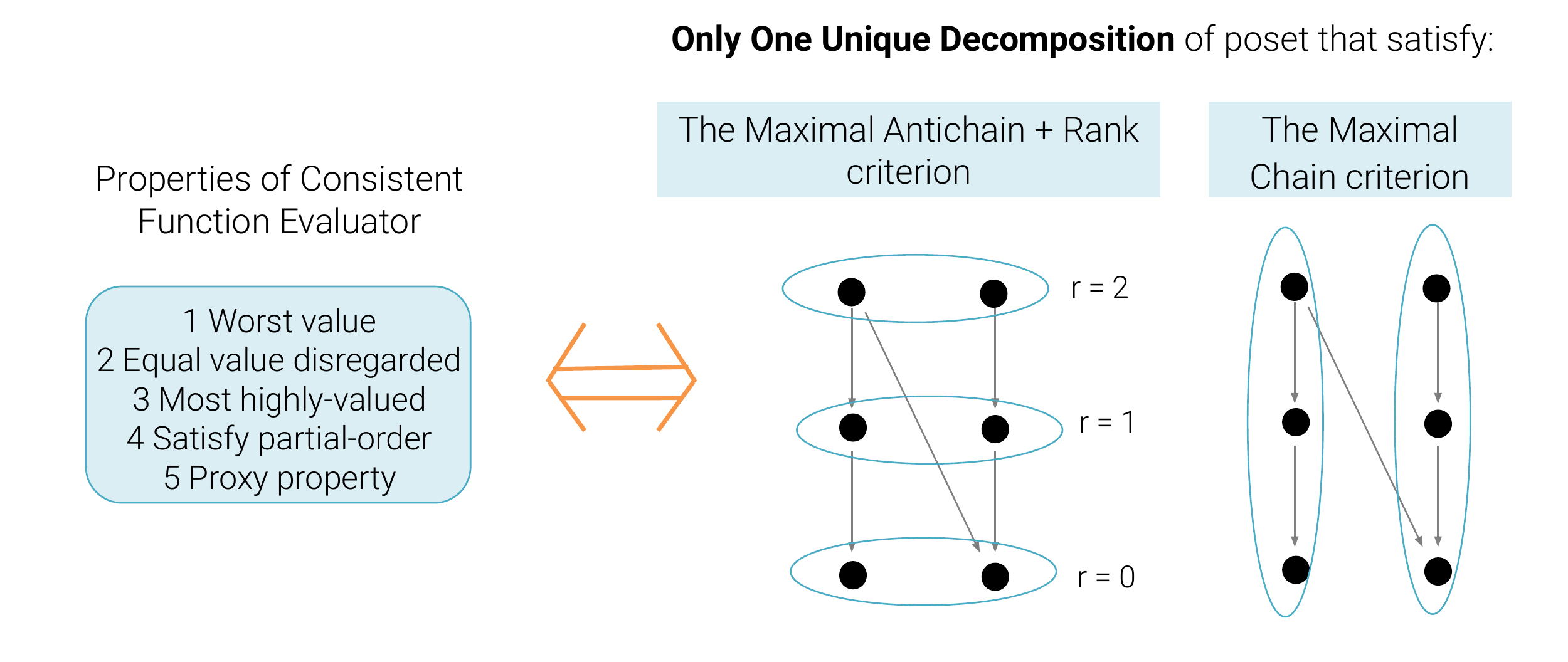}
      \caption{A visualization of Theroem~\ref{consist_evaluable}.}
      \label{proof_sketch}
   \end{figure}

\begin{definition}[Specification structure]
A specification structure is a finite, graded, partially ordered set of dimensional properties $\mathcal{P}$. Namely, if $\preceq$ is the ordering relation for $\mathcal{P}$ and $\prec$ is the strict version thereof satisfying $x \prec y \Leftrightarrow (x \preceq y \land x \neq y)$  then there exists a ranking function $\rho: \mathcal{P} \to \mathbb{N}$ such that
\begin{enumerate}
    \item $p_1 \prec p_2 \Rightarrow \rho(p_1) < \rho(p_2)$.
    \item $p_1 \lessdot p_2 \Rightarrow \rho(p_2) = \rho(p_1) + 1$.
    \item $p$ is a minimal element of $\mathcal{P} \Rightarrow \rho(p) = 0$.
\end{enumerate}
where $\lessdot$ denotes the \textit{covering relation} on $\mathcal{P}$ that satisfies 
\begin{align*}
p_1 \lessdot p_2 \Leftrightarrow p_1 \prec p_2 \land \forall p \in \mathcal{P}. \lnot(p_1 \prec p \land p \prec p_2)
\end{align*}
\end{definition}
We immediately have the following corollary.
\begin{corollary}
Any specification structure can be consistently evaluated.
\end{corollary}
\begin{proof}
This follows directly from Theorem~\ref{consist_evaluable} and the fact that any graded poset has properties 1) and 2) defined therein.
\end{proof}
The relation between consistently-evaluable sets and partially-ordered sets are shown more clearly in Fig. \ref{graded_posets_justify}. The main difference between consistently-evaluable posets and graded posets is the constraint that graded posets are defined such that the ranks assigned to two nodes which are comparable have to have a difference of one. It is easier to check whether a set is graded, as opposed to consistently-evaluable because of the following lemma: 
\begin{lemma}
\label{maxchain}
A poset is graded if and only if all of its maximal chains have the same length.
\end{lemma}
Any consistently evaluable poset can be reduced to a ``canonical'' form that has the graded property with the same consistent evaluation. 
\begin{theorem}
Each consistently evaluable poset can be turned into a graded poset that is equivalent under consistent evaluation.
\end{theorem}

   \begin{figure}[thpb]
      \centering
      \includegraphics[scale=0.32]{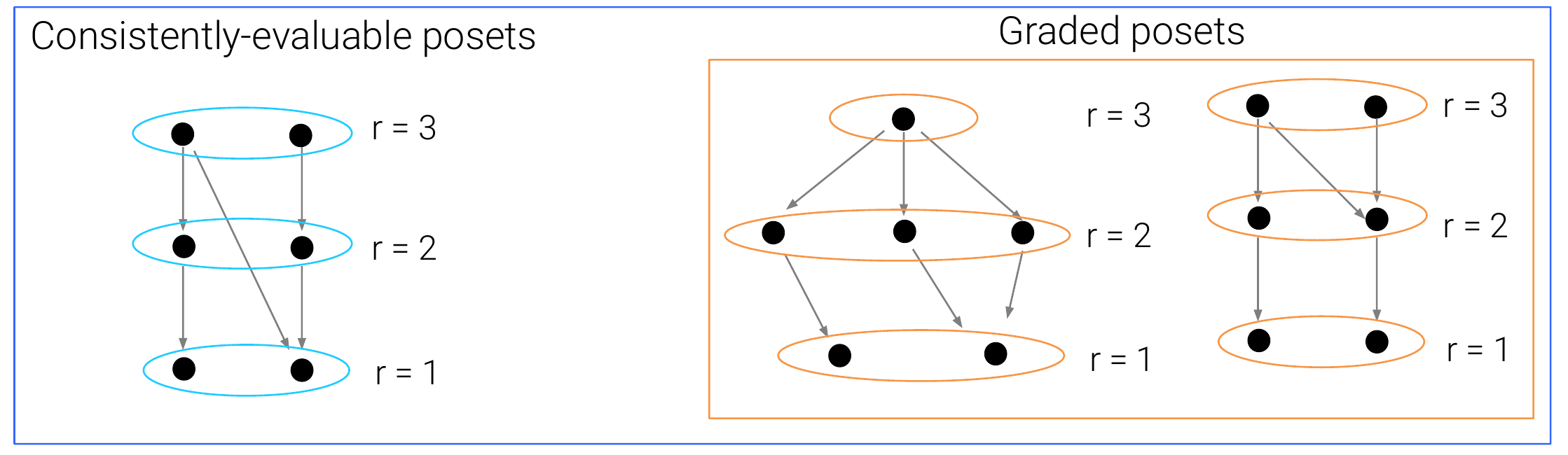}
      \caption{Graded posets (specification structures) are a subset of consistently-evaluable posets.}
      \label{graded_posets_justify}
   \end{figure}
   
\renewcommand{\proof}{Proof:}
\begin{proof}
This is achieved by removing all ``edges'' that span more than 2 levels of antichains in the unique partition of Theorem~\ref{consist_evaluable}. One can without much difficulty verify that doing so will remove all maximal chains with length strictly less than the total number of these antichains, which by Lemma~\ref{maxchain} implies that the resulting poset is graded. Since the other antichains are not affected by these operations, the resulting evaluation is not affected either.
\end{proof}
\renewcommand{\proof}{Proof:}

   \begin{figure}[thpb]
      \centering
      \includegraphics[scale=0.22]{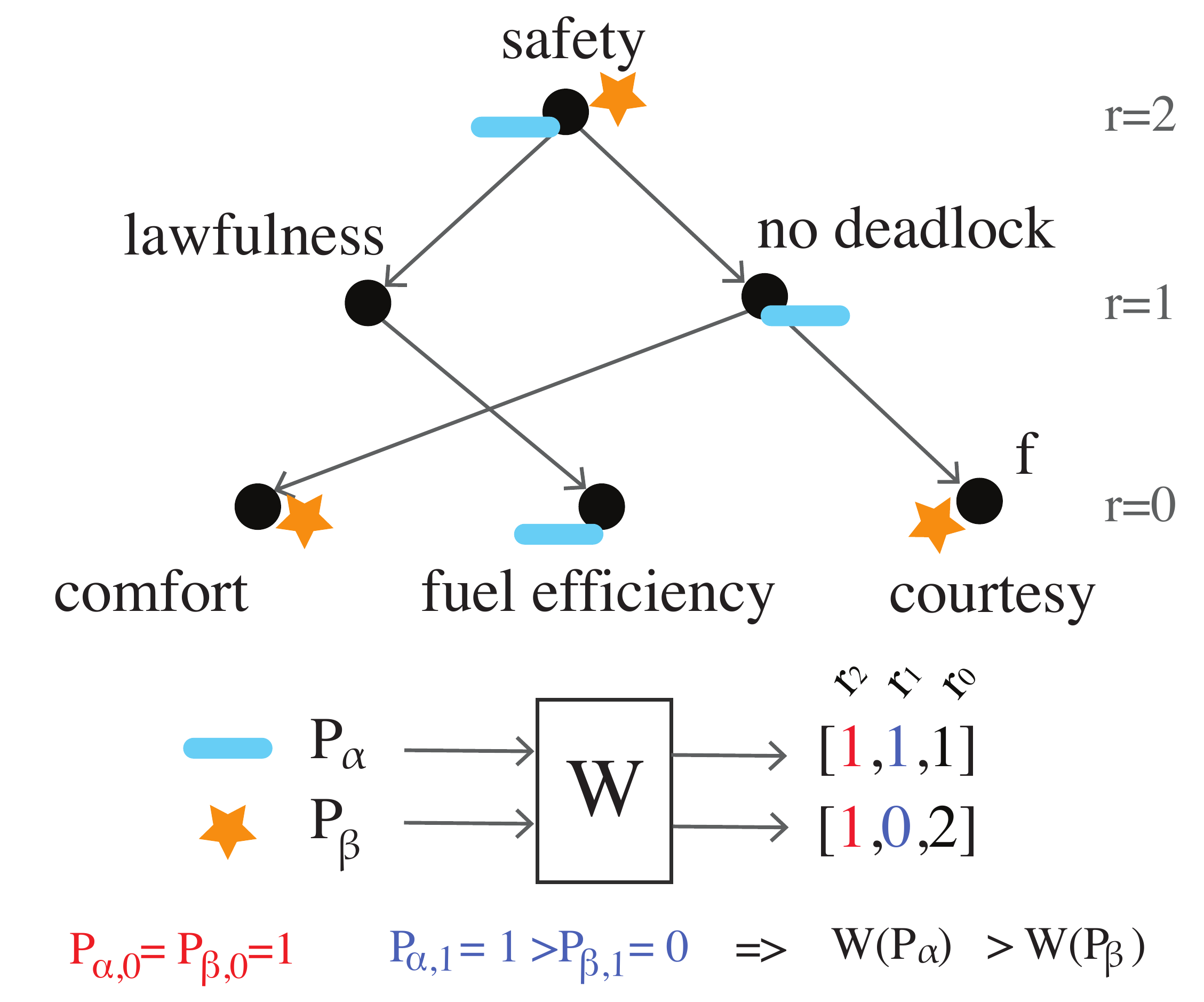}
      \caption{This shows how the consistent evaluator function $W$ works on a specification structure. The function W computes a tuple for each subset, and compares the elements from most significant to least significant digits (left to right).}
      \label{Wfunc}
   \end{figure}

\begin{example}[Evaluating a specification structure]
\label{W_func_spec}
Let $S$, $L$, $ND$, $FE$, $Cf$, $C$ be dimensional properties denoting safety, lawfulness, no deadlock, fuel efficiency, comfort, and courtesy respectively.
Let $P \coloneqq \set{S, L, ND,FE,Cf,C}$. The partial order on these dimensional properties is shown in Fig.~\ref{Wfunc}. Given the current world configuration, we assume the oracle can determine which subset of specifications will be satisfied by taking a given action. Let $P_{\alpha} \coloneqq \set{S, ND, FE}$  denote the subset of specifications satisfied by taking action $\alpha$. Similarly, let $P_{\beta} \coloneqq \set{S,Cf,C}$. To compare the actions $\alpha$ and $\beta$, given $P_{\alpha}$, $P_{\beta}$, we use the evaluator $W$ defined in the proof sketch of Theorem~\ref{consist_evaluable} to make the comparison. $W(P_{\alpha}) = [1, 1, 1]$ since there is one specification from each rank that can be satisfied by taking the action $\alpha$, and $W(P_{\beta}) = [1,0,2]$ since there is one property with highest rank (r=2) and two properties with lowest rank (r=0) that can be satisfied by taking action $\beta$. Therefore, to evaluate their relative importance, the most significant figure corresponds to the left-most element in the tuple since that element has the highest rank. We begin our comparison there. Note $W_i$ represents the $i$\/th the element of the tuple. Since $W_0(P_{\alpha}) = 1$ and $W_0(P_{\beta}) = 1$, we have to keep comparing elements in the tuple to determine which one has higher ordering. We find $W_1(P_{\alpha}) > W_1(P_{\beta})$. Therefore, $P_{\alpha}$ dominates $P_{\beta}$ by the weak order imposed by the $W$ evaluator and therefore, the action $\alpha$ should be chosen over $\beta$.
\end{example}
\medskip\noindent 
\section{Consistency and completeness}
We would like to introduce the ideas of consistency and completeness of a set of specifications that are hierarchically-ordered in a specification structure. In this context, consistency is defined as the ability to be uniquely and consistently evaluable. Completeness is defined by the extent of the dimensional properties specified. A specification structure that has more dimensional properties (i.e. encompasses a broader range of specifications) is therefore more complete. 
\subsection{Consistency} 
The notion of consistency comes from Theorem~\ref{consist}, which says that there is a unique weak order on the powerset of a specification structure regardless of the consistent evaluator used.
\begin{theorem}[Consistency implies uniqueness]
\label{consist}
If $\mathcal{P}$ is a poset with an ordering relation $\preceq$ that can be consistently evaluated, then all consistent evaluators of $\mathcal{P}$ are equivalent. That is, for any pair of consistent evaluators $f_a, f_b$ of $\mathcal{P}$, for all $P_1, P_2 \subseteq \mathcal{P}$, we have $$f_a(P_1) \leq f_a(P_2) \Leftrightarrow f_b(P_1) \leq f_b(P_2)$$
\end{theorem}
\renewcommand{\proof}{Proof (sketch):}
\begin{proof}
This result follows from applying Theorems~\ref{consist_evaluable} and~\ref{uniqueness} and Definition~\ref{evaluator}. A full proof is available in the Appendix.
\end{proof}
\renewcommand{\proof}{Proof:}

\subsection{Completeness}
As a dual to how the placement of the \textit{edges} of the directed graph presenting a specification structure determines its consistency, we propose that the inclusion (or exclusion) of \textit{nodes} determines its (relative) completeness. 

In order to make an existing specification structure more complete, we must be able to refine the graph in a consistent manner. Refinement is equivalent to adding dimensional properties (nodes) or comparisons (edges) to the specification structure in a way that preserves the gradedness property of a specification structure. We now define how to properly add a node or edge into the specification structure in a way that preserves the specification structure's mathematical properties.

The following is a direct corollary of Lemma~\ref{maxchain}.
\begin{corollary}[Proper mode or edge refinement]
\label{graphRefinement}
If a node (or an edge) is added to the specification structure such that its relationship to the other nodes (the comparison it makes) is defined in a way that all maximal chains 
have the same length, then the resulting partially ordered set is also a specification 
structure. 
\end{corollary}

Examples for proper (and improper) ways of adding a new node as well as examples of making minimal modifications to accommodate for an inconsistently-added node are included in the Appendix.

\begin{example}
Here, we give a very simple specification structure: lawfulness ($L$) $\prec$ no deadlock ($ND$) $\prec$ safety ($S$). We consider the consistent evaluator $W$ (presented in the proof sketch of Theorem~\ref{consist_evaluable}). 
$W$ will have the ordering $W(\{L\}) < W(\{ND\}) < W(\{S\}) < W(\{S, L\}) < W(\{S, ND\}) < W(\{S, L, ND\})$. The ordering intuitively means that a car should always prioritize taking actions that satisfy all three types of specifications. However, if there is a situation where a car cannot ensure safety without breaking the law, then it should break the law to maintain safety since $W(\{S\}) > W(\{L\})$. Also, this hierarchy says if there is a situation where the car is in a deadlock, it can break the law since $W(\{S, ND\}) > W(\{S, L\})$ as long as the action is still safe. 
\end{example}
\medskip\noindent 
As long as the car chooses behaviors that respect the weak order from the consistent evaluator on the specification structure, the system will satisfy the guarantees part of the assume-guarantee contract, and therefore perform actions that are ``correct''. We now introduce how assumptions can be defined with respect to the specification structure.
\section{Assume-guarantee profiling}
While each autonomous vehicle should only guarantee that it will behave according to a single specification structure, we want our assumptions on the environment (other agents) to accommodate for the diverse behaviors displayed by human drivers who may not follow the law all the time. This implies that other agents might choose to follow any one of a large number of possible specification structures.
   \begin{figure}
      \centering
      \includegraphics[scale=0.28]{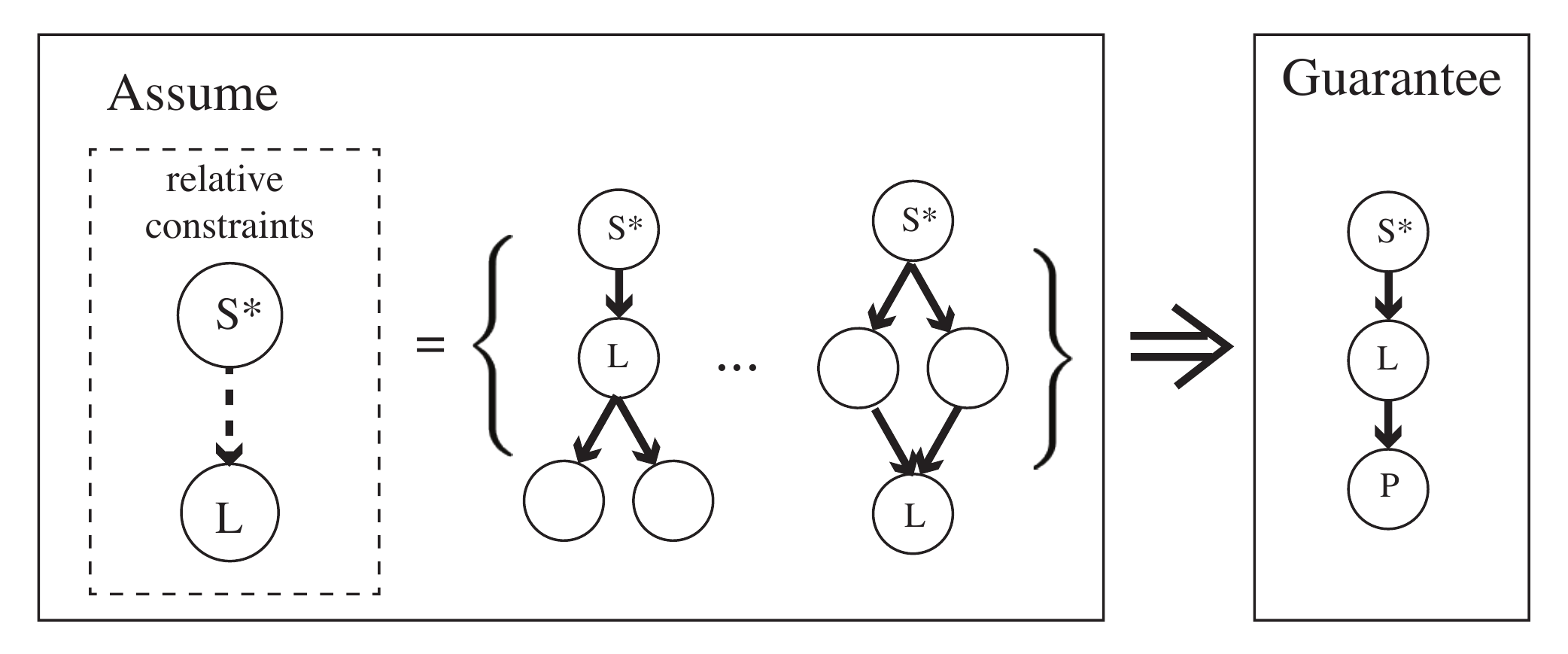}
      \caption{The assumptions are based on a set of specifications structures that satisfy some constraints, which is shown on the left. The guarantees are based on a single specification structure that is shown on the right.}
       \label{AGspecContract}
   \end{figure}
We constrain the set of specifications structures of other agents to always prioritize safety first. Since other agents presumably follow the law most of the time, we also include a relative ordering constraint where safety is prioritized before the law. We have only defined a relative ordering between safety and law in the assumptions since we do not exactly know where other dimensional properties will fit within that agent's specification structure. Therefore, our assumptions on the environment can be defined as follows: 
\begin{definition}[Assumption set]
\label{assumeSet}
Let $S$ denote the set of all specification structures. Let $P$ be a set of dimensional properties. Let $p \in P$. The assumptions in the assume-guarantee contract is defined as:
\begin{align*}
A_{\text{spec}} = \{S_i \in S|(\text{safety} \in S_i)\land (\text{lawfulness} \in S_i) \\ \land (\forall p \in S_i.p \preceq \text{safety}) \}.
\end{align*}
It is the set of all specification structures that both safety and lawfulness are included in the specification structure and that safety has the highest rank out of all dimensional properties included in the specification structure. 
\end{definition}

The following revised assume-guarantee definition of Definition~\ref{ag_profile}
characterizes the set of specification structures agents in the environment can be assumed to have and the specifications that an individual self-driving car can guarantee. 
\begin{definition}[Assume-guarantee profiling revised]
An assume-guarantee contract $\mathcal{C}$ defined for an agent is a pair $(\A,\G)$, where
\begin{enumerate}
    \item $\A$ is a set of specification structures for the agent's environment that is a subset of the set generated by Definition~\ref{assumeSet}.
    \item $\G$ is the guarantee that the agent with respect to a single, pre-defined specification structure.
\end{enumerate}
\end{definition}
This assume-guarantee profiling is shown in Fig.~\ref{AGspecContract}.
Let $\mathcal{J}$ be the index set for a set of agents. Given $\C_j = (\A_j, \G_j)$, where $j$ is the index of an agent and $\A_j$ are the assumptions that agent $j$ is making about its environment while $\G_j$ is its guarantees. We say that the group of agents indexed by $\mathcal{J}$ are \textit{compatible} if
\begin{align*}
\forall j \in \mathcal{J}.\forall i \in \mathcal{J}-\set{j}.\G_j \subseteq \A_{i}
\end{align*}
This says that the guarantees of agent $j$ must be included in the assumptions of all other agents in the compatible set. If one agent $i$ has guarantees corresponding to a specification structure that is not included in another agent $k$'s assumptions, then correct behavior cannot be guaranteed. 
Assuming that all agents' assumptions and guarantees are compatible, we can formulate the notion of a \textit{blame-worthy} action/strategy.
\begin{definition}[Blameworthy action]
A blameworthy action/strategy is one in which an agent violates its guarantees, thereby causing another agent's assumptions not to be satisfied and thus resulting in an unwanted situation where blame must be assigned.
\end{definition}
In order to show an example of an assume-guarantee contract that might be legally imposed for self-driving cars, we present a \textit{set of axioms for the road}. The specification structures defined in the assumptions and guarantees of this contract are intentionally left unrefined, since it would ultimately be up to a car-manufacturer to determine the remaining ordering of specification properties.
\begin{enumerate}
\item[A1] \textit{Other agents will not act such that collision is inevitable.}
\item[A2] \textit{Other agents will often act corresponding to traffic laws, but will not always follow them.}
\item[G1]  \textit{An agent will take no action that makes collision inevitable.} 
\item[G2]  \textit{An agent will follow traffic laws, unless following them leads to inevitable collision.} 
\item[G3]  \textit{An agent may violate the law if that can safely get it out of a dead/live-lock situation.}
\end{enumerate}

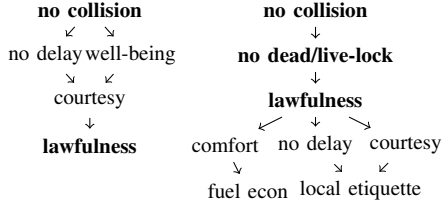
\begin{figure}
\centering
\begin{tikzpicture}[scale=.30]
  \node (safety_assume) at (-10,2) {\footnotesize$\textbf{no collision}$};
  \node (delay_assume) at (-11.8,0) {\footnotesize$\text{no delay  }$};
    \node (comfort_assume) at (-8.2,0)
    {\footnotesize$\text{well-being}$};
    \node (courtesy_assume) at (-10,-2)
    {\footnotesize$\text{courtesy}$};
\node (law_assume) at (-10,-4) {\footnotesize$\textbf{lawfulness}$};
\draw [->] (safety_assume) -- (delay_assume);
\draw [->] (delay_assume) -- (courtesy_assume);
\draw [->] (safety_assume) -- (comfort_assume);
\draw [->] (comfort_assume) -- (courtesy_assume);
\draw [->] (courtesy_assume) -- (law_assume);

  \node (safety) at (0,2) {\footnotesize$\textbf{no collision}$};
  \node (progress) at (0,0) {\footnotesize
 $\textbf{no dead/live-lock}$};
  \node (law) at (0,-2) {\footnotesize$\textbf{lawfulness}$};
  \node (comfort) at (-4,-4) {\footnotesize$\text{comfort}$};
  \node (delay) at (0,-4) {\footnotesize$\text{no delay}$}; 
    \node (courtesy) at (4,-4) {\footnotesize$\text{courtesy}$}; 
    \node (fuel) at (-3,-6) {\footnotesize$\text{fuel econ}$}; 
    \node (culture) at (2,-6) {\footnotesize$\text{local etiquette}$}; 
    
\draw [->] (safety) -- (progress);
\draw [->] (progress) -- (law);
\draw [->] (law) -- (comfort);
\draw [->] (law) -- (delay);
\draw [->] (law) -- (courtesy);
\draw [->] (comfort) -- (fuel);
\draw [->] (delay) -- (culture);
\draw [->] (courtesy) -- (culture);
\end{tikzpicture}
\caption{Two examples of refined assumption (left) and guarantee (right) specification structures. Dimensional properties of the root structures are in bold text. The left could represent an ambulance and the right could represent a civilian vehicle.}
\label{axguarantees}
\end{figure}
We can see from Fig.~\ref{axguarantees} how these axioms have a direct mapping to a specification structure. We argue that this sort of root structure might be imposed by a governing body to ensure the safe behaviors of self-driving cars. 

\section{Game Examples}
In this section, we present some preliminary examples of how these types of high-level behavioral specifications that are defined via these specification structures, might be applied in some traffic scenarios. 

Under the simplified assumption that each agent has a single specification structure (i.e. agents are not human), each agent will have a well-defined ordering of which actions have higher value, and will therefore have a well-defined utility function over actions. Game theory provides a mathematical model of strategic interaction between rational decision-makers that have known utility functions\cite{gameTheory}. We can therefore use game-theoretic concepts to analyze which pair of actions will be jointly advantageous for the agents given their specification structures. 

\begin{example}
Consider the case where there are two agents, each of whose specification structures are specified in Fig.~\ref{game1}. In this game, Player Y encounters some debris, and must choose an action. Player Y can either choose to stay in its current location, or do a passing maneuver that requires it to break the law. Player X represents a car moving in the opposite direction of Player Y. In this case, Player X can either move at its current velocity or accelerate. The move and accelerate action make Player X move one and two steps forward respectively. The $W$ function is the same as one the proof sketch of Theorem~\ref{consist_evaluable}.
   \begin{figure}[thpb]
      \centering
      \includegraphics[scale=0.38]{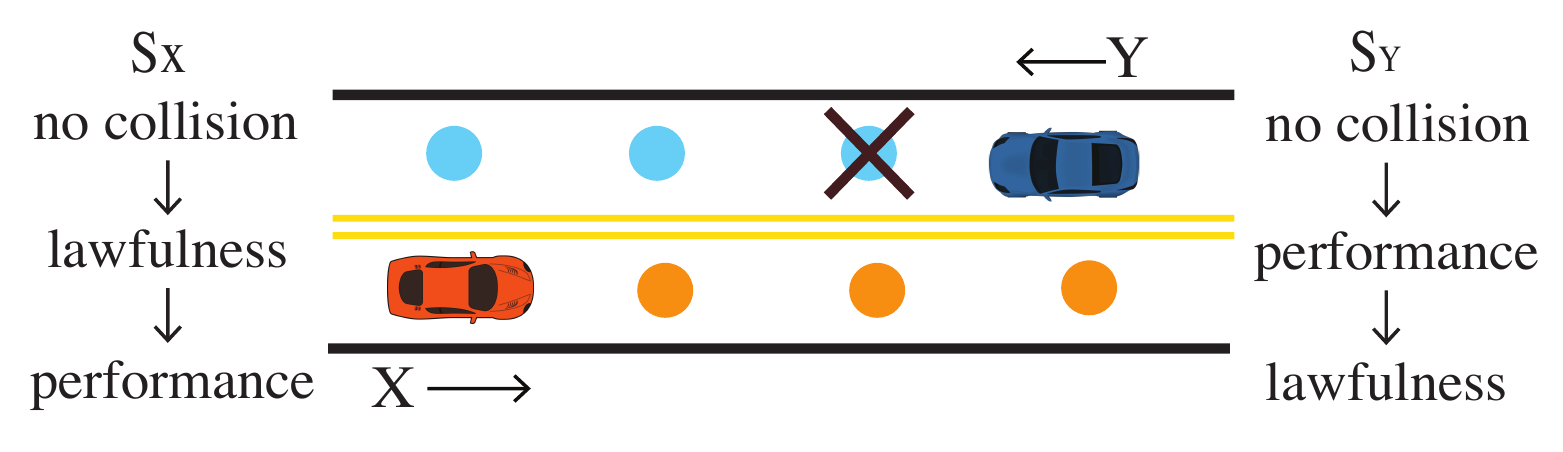}
      \caption{The game scenario when Player Y encounters debris on its side of the road. The specification structures of each of the agents are given by $S_x$ and $S_y$.}
      \label{game1}
   \end{figure}
$W_x$ is evaluated on the specification structure $S_x$ shown on the left side of Fig.~\ref{game1} and $W_y$ is evaluated on $S_y$. Assuming there is a competent oracle who gives the same predictions for both agents, the resulting payoff matrix according to the specification structures are given in Table \ref{game1strat} (note that an equivalent decimal conversion of the scores is given for ease of reading).
\begin{table}[H]
\centering
\caption{}
\label{game1strat}
\begin{tabular}{c|c|c}
playerX/playerY & Stay & Pass \\
\hline
Move & \begin{tabular}[c]{@{}c@{}}$W_x(1,1,0) \sim 6$\\ $W_y(1,0,1) \sim 3$\end{tabular} & \begin{tabular}[c]{@{}c@{}}$W_x(1,1,0) \sim 6$\\ $W_y(1,1,0) \sim 6$\end{tabular} \\
\hline
Accelerate & \begin{tabular}[c]{@{}c@{}}$W_x(1,1,1) \sim 7$\\ $W_y(1,0,1) \sim 3$\end{tabular} & \begin{tabular}[c]{@{}c@{}}$W_x(0,0,0) \sim 0$\\ $W_y(0,0,0) \sim 0$\end{tabular}
\end{tabular}
\end{table}
From the table, we can see there are two Nash equilibria in this game scenario. The two equilibria are Pareto efficient, meaning there are no other outcomes where one player could be made better off without making the other player worse off. Since there are two equilibria, there is ambiguity in determining which action each player should take in this scenario despite the fact that the specification structures are known to both players.  
\end{example}
\medskip\noindent 
There is a whole literature on equilibrium selection \cite{gameTheory}. The easiest way to resolve this particular stand-off, however, would be to either 1) communicate which action the driver will take or 2) define a convention that all self-driving cars should have when such a situation occurs. In this particular scenario, however, Player X can certainly avoid accident by choosing to maintain speed while Player Y can also avoid accident by staying. If any ``greedy'' action of either Player X or Y would pose the risk of crashing depending on the action of the other player. This suggests a risk-averse resolution in accident-sensitive scenarios like this one. We will focus on defining a more systematic way of resolving multiple Nash equilibria in future work.
\begin{example}
For this paper, we have abstracted the perception system of the self-driving car to the all-knowing oracle. We first consider the case where the oracles on each of the cars are in agreement, and then consider the potential danger when the oracles of the cars differ. In this scenario, we assume that there are two cars that are entering an intersection with some positive velocity, as shown in Fig.~\ref{game2}. 

    \begin{figure}[thpb]
      \centering
      \includegraphics[scale=0.38]{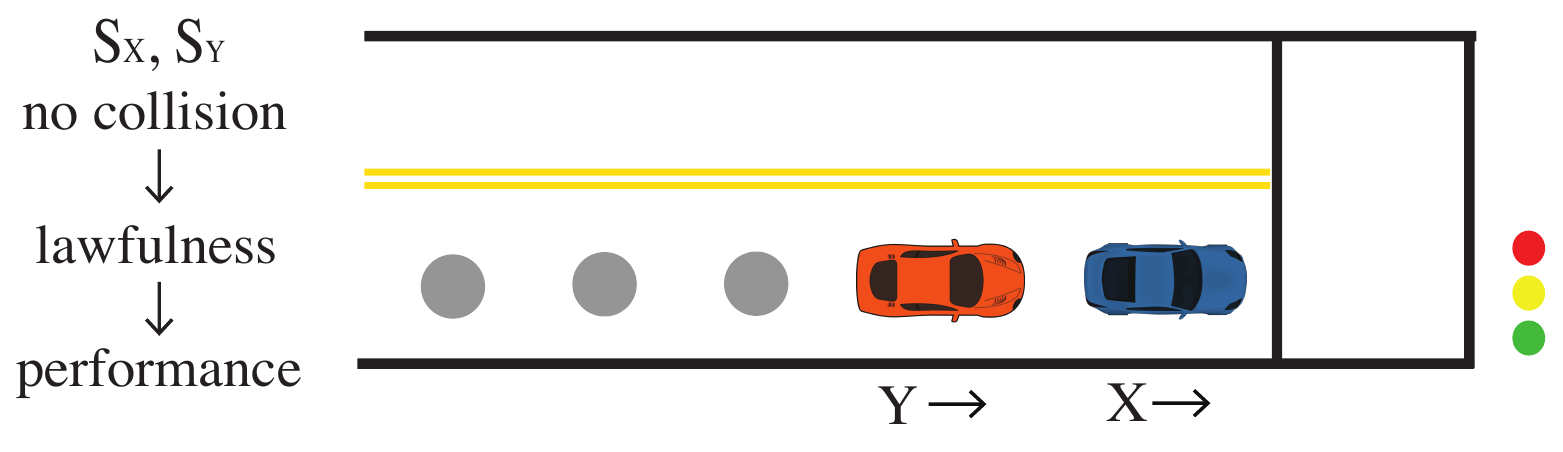}
      \caption{The game scenario where two cars are approaching an intersection, but have different beliefs about the state of the traffic light.}
      \label{game2}
   \end{figure}

In the case where both vehicles' oracles agree on the same information, i.e. that the yellow light will remain on for long enough for both vehicles to move past the intersection, the best action for both Player X and Player Y is to move forward. 

Now, consider the case where the oracles are giving incompatible beliefs about the environment, namely, the state of the traffic signal. Let X have the erroneous belief the traffic light will turn red very soon, and it assumes that Y's oracle is believes the same thing. X's oracle gives rise to Table~\ref{bad1}, according to which, the conclusion that Player X will make is that both of the cars should choose to slow down.  

\begin{table}
\centering
\caption{}
\begin{tabular}{c|c|c}
\label{bad1}
playerX/playerY & Slow & Move \\
\hline
Slow & \begin{tabular}[c]{@{}c@{}}$W_x(1,1,0) \sim 6$\\ $W_y(1,1,0) \sim 6$\end{tabular} & \begin{tabular}[c]{@{}c@{}}$W_x(0,1,0) \sim 2$\\
$W_y(0,0,0) \sim 0$\end{tabular} \\
\hline
Move & \begin{tabular}[c]{@{}c@{}} $W_x(1,0,1) \sim 3$\\ $W_y(1,1,0) \sim 6$ \end{tabular} & \begin{tabular}[c]{@{}c@{}}$W_x(1,0,1) \sim 3$\\ $W_y(1,0,1) \sim 3$\end{tabular}
\end{tabular}
\end{table}

Assume that Y has a perfect oracle that predicts the traffic light will stay yellow for long enough such that Y would also be able to make it through the intersection. If Y assumes that X has the same information (see Table~\ref{bad2}), then the best choice for both is to move forward into the intersection. 
\begin{table}[H]
\centering
\caption{}
\begin{tabular}{c|c|c}
\label{bad2}
playerX/playerY & Slow & Move \\
\hline
Slow & \begin{tabular}[c]{@{}c@{}}$W_x(1,1,0) \sim 6$\\ $W_y(1,1,0) \sim 6$\end{tabular} & \begin{tabular}[c]{@{}c@{}}$W_x(0,1,0) \sim 2$\\ $W_y(0,0,0) \sim 0$\end{tabular} \\
\hline
Move & \begin{tabular}[c]{@{}c@{}}$W_x(1,1,1) \sim 7$\\ $W_y(1,1,0) \sim 6$\end{tabular} & \begin{tabular}[c]{@{}c@{}}$W_x(1,1,1) \sim 7$\\ $W_y(1,1,1) \sim 7$\end{tabular}
\end{tabular}
\end{table}
The incompatible perception information will thus cause Player X to stop and Player Y to move forward, ultimately leading to collision. 

This particular collision is caused by errors in the perception system. Future work will need to focus on developing a better perception system or on creating a system that will yield correct behaviors even perception uncertainty. 
\end{example}
\section{Conclusion and Future Work}
To summarize, we have introduced a framework that allows us to formulate specifications that govern high-level behaviors of autonomous vehicles. If specifications are hierarchically ordered into a specification structure, actions and strategies can be compared to one another in a consistent manner. Furthermore, the specification structure can be made more complete by properly defining new properties and relations. We introduce the idea of having assume-guarantee contracts defined over these specification structures that serve as profiles for implicit agreement. A contract essentially says if the environment can be assumed to behave according to some softly-constrained specification structure, then the self-driving car can guarantee it will behave according to its own specification structure. Blame is defined as the case where a car does not act according to its assumed specification structure. Finally, we provide some examples of how cars following these specification structures might behave in game-theoretic experiment settings. 

In the future, we hope to extend the game-theoretic framework that deals with two cars that know each other's specification structures to a case that deals with human-driven car with a more ambiguous specification structure and a self-driving car with a known specification structure. Integrating perception uncertainty in this formulation would make it more applicable to real-life settings. Another interesting direction is to investigate sufficient and/or necessary conditions for stronger guarantees (such as robust safety) under the profiling framework by combining tools and techniques from formal methods (e.g., inductive invariants) and control theory (e.g., Lyapunov functions).


\bibliographystyle{plain}
\bibliography{refs}
\nocite{*}
\newpage
\clearpage
\newpage

\newpage
\clearpage
\newpage
\section*{Appendix}
\subsection{Examples of Consistent Evaluators}
\begin{figure}[H]
    \centering
   \includegraphics[scale=0.45]{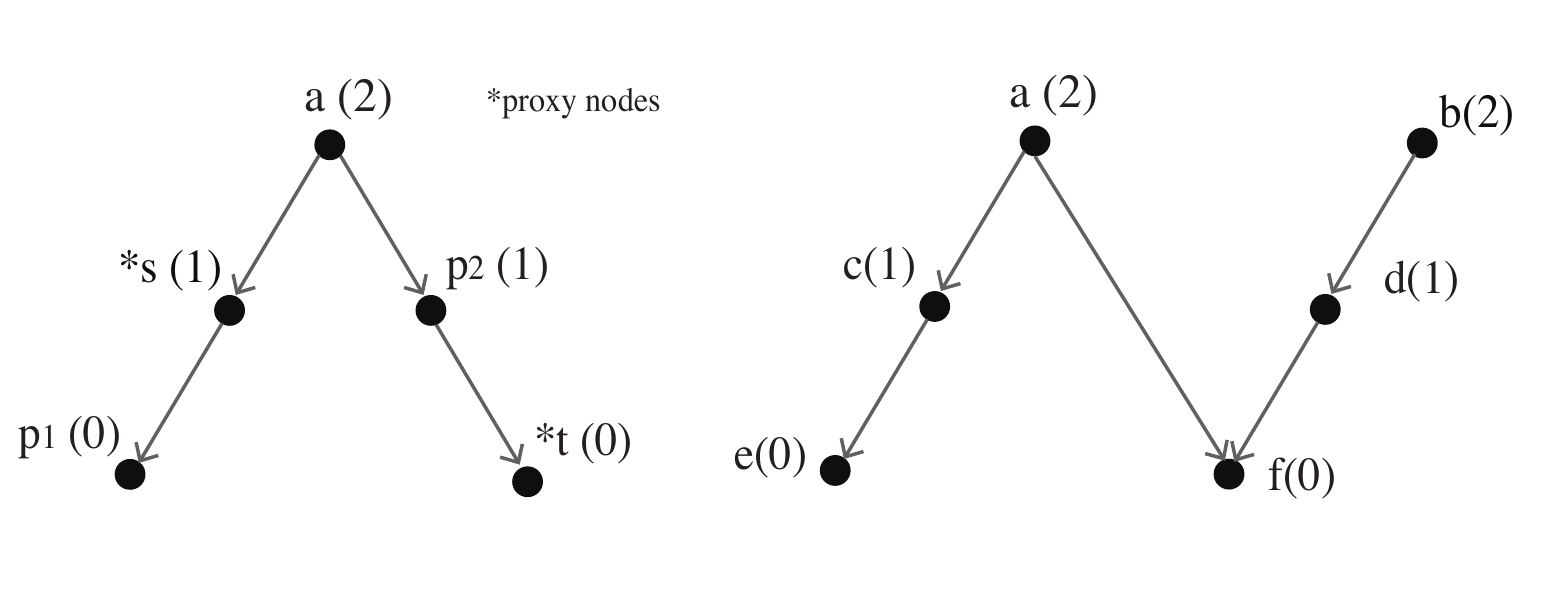}
    \caption{Both posets admit consistent evaluators. The value the consistent evaluator assigns on each singleton that consists of the node is written in parentheses. Note that the poset on the left has the additional graded property while the one on the right does not.}
    \label{consistent}
\end{figure}

Consider the posets in Fig.~\ref{consistent}. Any evaluator $f$ that assigns the values according to the values in the parentheses, shown in the figure, can easy be verified to have properties~\ref{worstbest}-\ref{evalorder} of consistent evaluation. On the left poset, we can also see that property~\ref{evalcaution} is also satisfied since for every pair of nodes such that $f(\set{p_1}) < f(\set{p_2})$ implies there exist proxy nodes $s$ and $t$ such that $f(\set{p_1}) = f(\set{t})$, $f(\set{p_2}) = f(\set{s})$ and $p_1 \prec s$ and $p_2 \prec t$. This relation can be easily seen for the nodes $p_1$ and $p_2$ in Fig.~\ref{consistent}. The same statement applies to the poset on the right, which is not graded, but also happens to be consistently evaluable!

\subsection{Proof of Theorem \ref{consist_evaluable}}
\begin{proof}
($\Rightarrow$):
Suppose that $P$ is a poset of dimensional properties with the 
ordering relation $\preceq$ such that $P$ has a consistent 
evaluator $f$. Since $P$ is finite, the set $f_P \coloneqq 
\set{f(\set{p}) \mid p \in P}$ is also finite. Furthermore, the
range of $f$ being totally ordered implies we can write $f_P = 
\set{z_1, z_2, \ldots, z_n}$ for $n = \abs{f_P}$ such that $z_1
< z_2 < \ldots < z_n$. For each $z_i \in f_P$, let $f^{-1}(z_i)
\subseteq P$ be defined by $f^{-1}(z_i) \coloneqq \set{p \mid p
\in P \land f(\set{p}) = z_i}$. Observe that requirement~\ref{evalorder} of Definition~\ref{evaluator} implies that for each $i$, $f^{-1}(z_i)$ is an antichain. Consequently, the $f^{-1}(z_i)$'s form a partition of $P$ by antichains. By ranking each $f^{-1}(z_i)$ by the corresponding $z_i$, it also follows that the antichains respect the partial order defined by $\preceq$. To show maximality, suppose that there exists $j \in [n]$ such that $f^{-1}(z_j)$ is not a maximal antichain. This implies that there exists $k \in [n] - \set{j}$ and there is a property $q^{\star} \in f^{-1}(z_k)$ such that $\set{q^{\star}} \cup f^{-1}(z_j)$ is an antichain. WLOG, suppose $j < k$ so that $z_j < z_k$ implies $\forall q \in f^{-1}(z_j). f(q) < f(q^{\star})$. But the existence of any $\tilde{q} \in f^{-1}(z_j)$ such that $f(\tilde{q}) < f(q^{\star})$ implies, by requirement~\ref{evalcaution} of Definition~\ref{evaluator}, that there exists $q' \in f^{-1}(z_j)$ such that $q' \prec q^{\star}$. But this contradicts the assumption that $\set{q^{\star}} \cup f^{-1}(z_j)$ is an antichain. From this, we conclude that \ref{rank_prop}) holds. To see that \ref{path_prop}) holds, observe that any property $p \in f^{-1}(z_j)$, if $j \neq n \land n \geq 2$, then by requirement~\ref{evalcaution} and the antichain property of the $f^{-1}(z_i)$, there exists $q \in f^{-1}(z_{j+1})$ such that $p \prec q$. Similarly, if $j \neq 1 \land n \geq 2$, there exists $r \in f^{-1}(z_{j-1})$ such that $r \prec p$. Applying this argument to $r$ and/or $q$ inductively yields a chain of length $n$ that contains $q$. This chain is maximal by the contradiction resulting from applying the pigeonhole principle to the assignment of properties from any chain of greater length to the maximal antichains.
\\
($\Leftarrow$): 
Let the maximal antichains in the partition of $P$ be $P_0,P_1,\ldots,P_{m-1}$ with ranks $r(P_0) < r(P_1) < \ldots < r(P_{m-1})$. 
We construct a function $W:2^P \to \mathbb{N}^{m}$ as follows. For any subset $S \subseteq P$, we define $A_{S,r}$ to be the set of all elements in the subset $S$ with rank $r$.
$W(S) \in \mathbb{N}^{m}$ the $(m)$-tuple whose $i$\/th element, or digit, $W_i(S) \coloneqq \abs{A_{S,i}}$. This means that the $i$\/th element in the tuple is the number of elements in the subset $S$ with rank $i$. 

The $i$\/th digit of $W(S)$ is defined to be more significant than the $j$\/th digit if the former is associated with a higher rank. This induces a natural total ordering relation $\leq$ on the set $\set{W(P) \mid P \subseteq \mathcal{P}}$ by most significant digits. In particular, this means, $W(S_a) \leq W(S_b)$ if and only if all corresponding entries are equal or the first most significant differing pair satisfies $W_i(S_a) < W_i(S_b)$. The rest of the proof involves checking that $W$ has all the properties of a consistent evaluator. 

We can easily verify that requirements \ref{worstbest}-\ref{evalorder} of a consistent evaluator are satisfied. Now, we show that requirement \ref{evalcaution} holds as well. Let us show this by contradiction. Consider there exists a node $p_1$ and $p_2$ such that $f(\{p_1\}) < f(\{p_2\})$ but there does not exist a node $s$ or $t$ such that $p_1 \prec s$, $t \prec p_2$, and  $f(\{p_2\}) = f(\{s\})$ and $f(\{p_1\}) = f(\{t\})$. WLOG, consider $p_1$ to be a node where there does not exist a node $s$ such that $p_1 \prec s$ and $f(\{p_2\}) = f(\{s\})$. Since there is no node that is directly comparable to $p_1$ in the antichain with value equal to $f(\{p_2\})$, there exists a maximal chain containing $p_1$ that has length strictly less than $m$. This is a violation of property \ref{path_prop}) characterizing the poset $P$.
\end{proof}

\subsection{Proof of Theorem~\ref{uniqueness}}
\begin{proof}
Suppose that $P_1, P_2, \ldots, P_m$ is also a partition of maximal antichains of $P$ with ranks $r(P_1) < r(P_2) < \ldots < r(P_m)$ that respect the partial order. Suppose that $m \neq n$ where $n = \abs{f_P}$. If $m > n$, then by \ref{path_prop}) of Theorem~\ref{consist_evaluable} there is a chain of length $m$. However, assigning these $m$ properties to the $f^{-1}(z_i)$ means by the pigeonhole principle that there are at least two properties that are assigned to the same $f^{-1}(z_j)$ for some $j$, implying that $f^{-1}(z_j)$ is not an antichain, a contradiction. It follows that $m \leq n$. Similarly, we can argue that $m \geq n$ and therefore $m = n$. Now, we claim that $P_i = f^{-1}(z_i)$ for all $i \in \set{1,2,\ldots, n}$. Suppose this is not the case, then there exists $p \in P_k$ such that $p \in f^{-1}(z_h)$ for $k \neq h$. Then by \ref{path_prop}), there are two chains of length $n$: $p_1 \prec p_2 \prec \ldots \prec p_n$ and $f_1 \prec f_2 \prec \ldots \prec f_n$ such that $p_i \in P_i$ and $f_i \in f^{-1}(z_i)$. We also have $p_k = f_h = p$. WLOG, assume $h < k$. This implies that $p_1 \prec p_2 \prec \ldots \prec p_k = p = f_h \prec f_{h+1} \prec \ldots \prec f_{n}$. However, this chain has length $k + n - h > n$ since $k > h$. This contradicts the fact that $P$ can be partitioned into $n$ antichains.
\end{proof}
\subsection{Proof of Theorem \ref{consist}}
\begin{proof}
By symmetry, it is sufficient to prove the 
($\Rightarrow$) direction. Suppose $f_a(P_1) 
\leq f_a(P_2)$. Now since $\mathcal{P}$ is consistently evaluable, by Theorems~\ref{consist_evaluable} and ~\ref{uniqueness}, it can be partitioned by a unique set of maximal antichains $\set{A_r}_{r=1}^{R}$. By 
requirement~\ref{evalorder} of Definition~\ref{evaluator}, one can show that any consistent evaluator will rank these antichains the same way. Namely, any consistent evaluator $f$ of $\mathcal{P}$ can be assumed, WLOG, to satisfy the following conditions
\begin{enumerate}
    \item\label{less} $f(\set{p_1}) < f(\set{p_2}) < \ldots < f(\set{p_R})$, for $p_r \in A_r, r \in \set{1,\ldots, R}$
\item\label{equal} $f(\set{q_i}) = f(\set{r_i})$, for any $q_r, r_r \in A_r, r \in \set{1,\ldots, R}$
\end{enumerate}
Condition~\ref{equal} above implies that all pairs of nodes that are of equal value to $f_a$ are also of equal value to $f_b$ and vice versa. So by requirement~$\ref{evalremove}$ of Definition~\ref{evaluator}, we can assume that $P_1$ and $P_2$ do not overlap in property values due to either $f_a$ or $f_b$. If $P_1 = \varnothing$, then by requirement~\ref{worstbest}, we have $f_b(P_1) < f_b(P_2)$. Otherwise, by requirement~\ref{evaldiscrim}, let $p_i^{\star} \in P_i$ be the property that maximizes the value of $f$ on $P_i$, we have $f_a(P_1) \leq f_a(P_2)$ implies that $p_1^{\star} \prec p_2^{\star}$ since $f_a(\set{p_1^{\star}}) \neq f_a(\set{p_2^{\star}})$ due to $P_1$ and $P_2$ not overlapping in property values and therefore $f_b(P_1) < f_b(P_2)$.
\end{proof}
\subsection{Adding Nodes to a Specification Structure}
Since car-manufacturers will want to refine a given root specification structure, it is important to define the ways nodes (or dimensional properties) can be added and still preserve the graded property of the specification structure. In Fig.~\ref{addingNodes}, we see some simple examples of how nodes can be added and the maximal chain property defined in Lemma \ref{maxchain} is maintained. 
    \begin{figure}[thpb]
      \centering
      \includegraphics[scale=0.26]{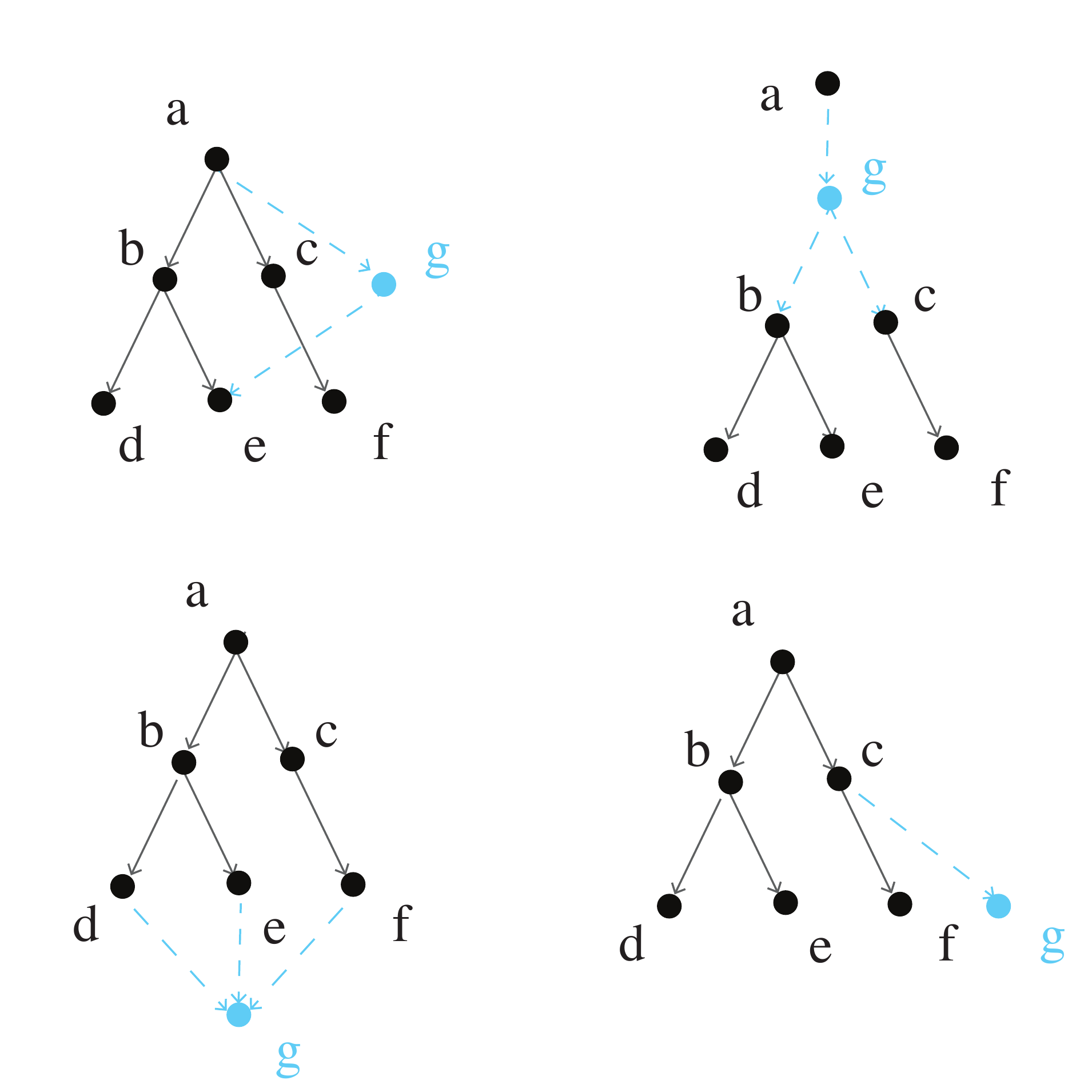}
      \caption{Shows different ways a dimensional property can be added to a graded poset and still preserve the graded property.}
      \label{addingNodes}
   \end{figure}
When a node is added to an anti-chain with a given rank $r$, it must be valued less than at least one node with rank $r+1$ and greater than at least one node of rank $r-1$ if such nodes exit. This is clearly displayed in the top-left example and bottom-right example in Fig.~\ref{addingNodes}. When a specification is added in-between existing rankings to create a new ranking, the node must be compared to all nodes in the rank above and the rank below in a manner that is consistent with the existing partial order. This can be seen in the top right and left bottom examples in Fig.~\ref{addingNodes}. 
Oftentimes comparisons of a new node  with existing nodes in a specification structure will result in a poset that is no longer graded. When the resulting poset is no longer graded, we introduce a way to make minimal modifications to the poset such that it regains its graded property. We mean the modifications are minimal in the sense that they do not significantly redefine the existing relationships between nodes.
    \begin{figure}[thpb]
      \centering
      \includegraphics[scale=0.26]{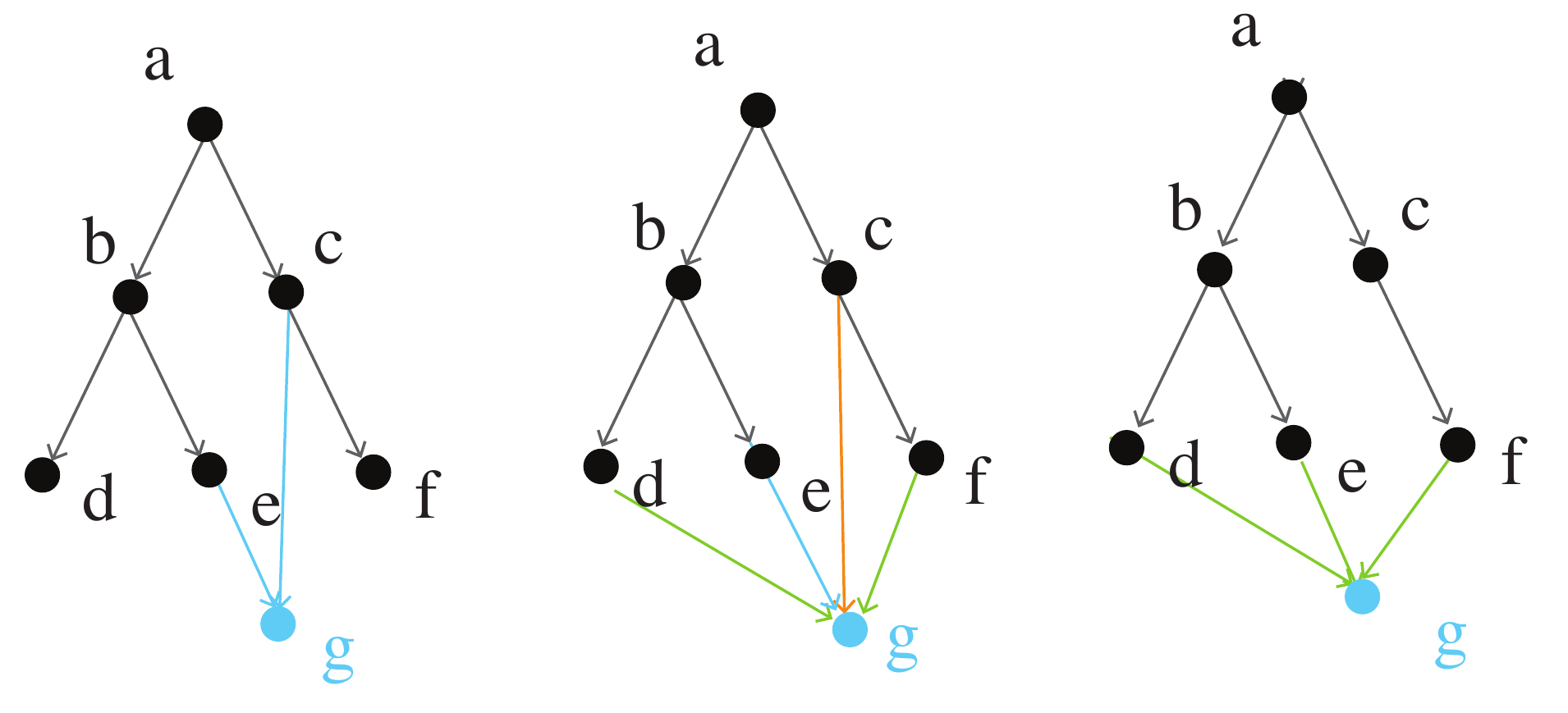}
      \caption{This shows the steps taken to add the node $g$ such that $g < e$ and $g < c$. The orange edges are deleted and the green edges are added to minimally change the poset to a graded poset.}
      \label{addingBadNodes}
   \end{figure}
In the particular scenario where a node is added such that it has a lower value than two nodes of ranks with a difference of one as shown in Fig.~\ref{addingBadNodes}, it is best to preserve the edge with the node in the poset with smaller rank, remove the edge of the node with higher rank and redefine that comparison via a proxy node. We can see it can be burdensome to exhaustively define all the different ways a node could be compared to existing nodes in a graded poset. Here are some general guidelines to follow when trying to add or remove edges to regain the graded property of the poset: 
\begin{enumerate}
    \item If an edge is redundant (i.e. the comparison is already defined via another node) then remove it. 
    \item Add edges between incomparable nodes of the poset. 
\end{enumerate}

\subsection{Adding Edges to a Specification Structure}
The same guidelines for resolving improperly added nodes applies to edges as well. Note that the only incomparable nodes are ones of the same rank. 

    \begin{figure}[thpb]
      \centering
      \includegraphics[scale=0.33]{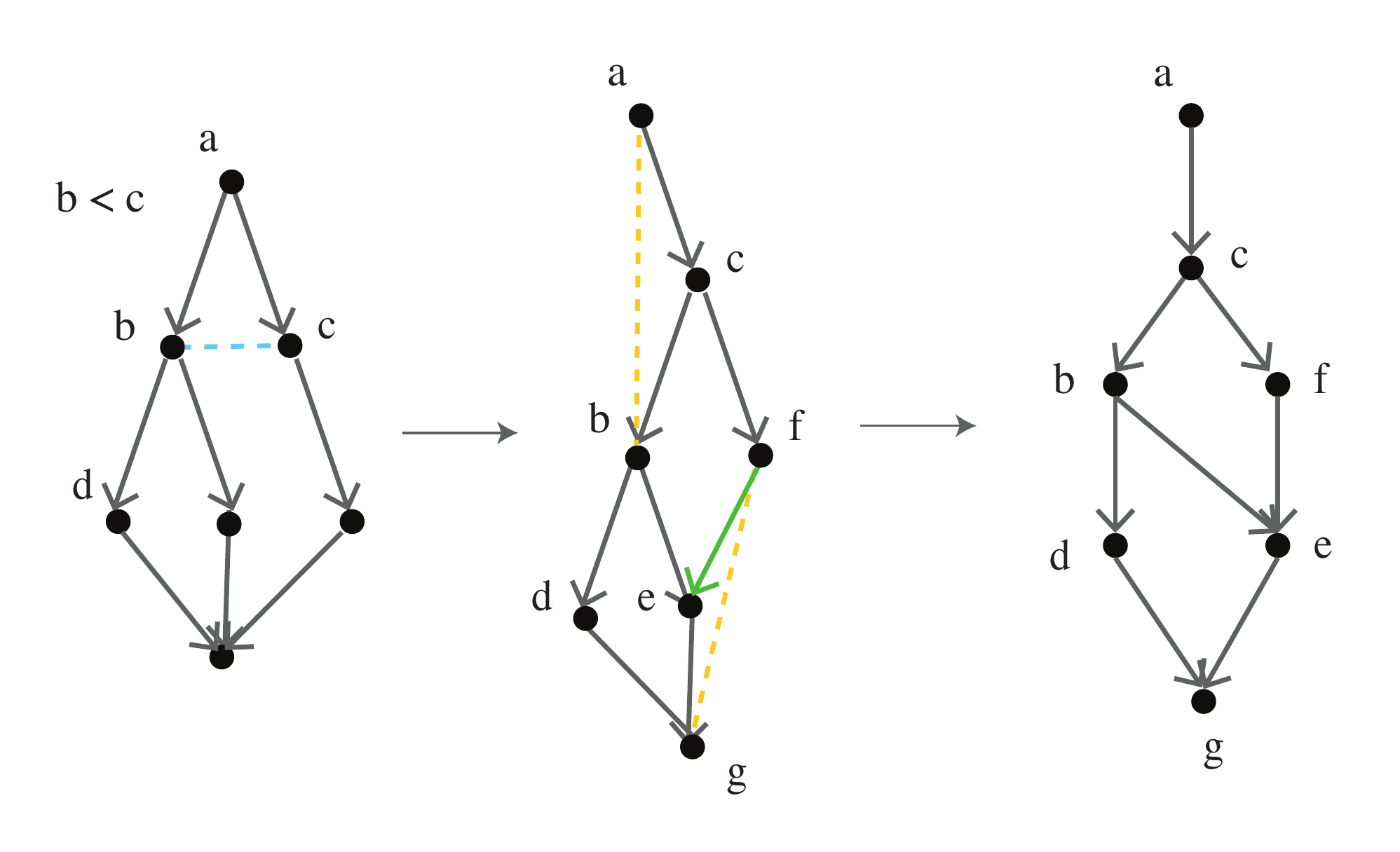}
      \caption{This shows the sequence of steps to add in the edge $b < c$ into the existing poset on the left. Orange edges are removed and green edges are added.}
      \label{addingEdges}
   \end{figure}
The example shown in Fig.~\ref{addingEdges} shows how adding a comparison between the nodes $b$ and $c$ causes the set to no longer have the graded property. By adding edges between incomparable nodes and also removing redundant edges, the gradedness of the poset returns.

\end{document}